\documentclass[journal]{IEEEtran}
\usepackage{amsthm}
\usepackage{amssymb}
\usepackage{graphicx}
\usepackage{amsmath}
\usepackage{newtxmath}
\usepackage[cal=cm]{mathalpha} 
\usepackage{tablefootnote} 
\usepackage{threeparttable} 
\usepackage{bm}
\newtheorem{theorem}{Theorem}{}
{}
{}
{}
{}
{}

\usepackage{algorithm}
\usepackage{algpseudocode}
\usepackage{graphics}
\usepackage{epsfig}
\usepackage{multirow}
\usepackage{caption}
\usepackage{array}
\usepackage{cite}
\usepackage{stfloats}
\usepackage{subfig} 
\usepackage{booktabs}
\usepackage{comment}
\usepackage{color}
\usepackage{tabu} 
\usepackage{stfloats}
\usepackage{mathtools}
\usepackage{stmaryrd} 
\makeatletter
\let\sum\relax
\let\prod\relax
\DeclareSymbolFont{largesymbols}{OMX}{cmex}{m}{n}
\DeclareMathSymbol{\sum}{\mathop}{largesymbols}{"50}
\DeclareMathSymbol{\prod}{\mathop}{largesymbols}{"51}
\usepackage[cal=cm]{mathalpha} 
\usepackage[colorlinks,linkcolor=blue,anchorcolor=blue,citecolor=blue,bookmarks=true]{hyperref}
\usepackage[capitalize]{cleveref} 
\crefname{figure}{Fig.}{Figs.}   
\Crefname{figure}{Fig.}{Figs.}   
\usepackage{subcaption}
\begin{document}
	\vspace{-1em}
	\title{ Joint Shape-Position Optimization Enhanced 2D DOA Estimation in Movable Antenna Systems}
	
	\vspace{-0.5em}
	\author{\IEEEauthorblockN{
			Chengzhi Ye, \IEEEmembership{Student Member,~IEEE}, 
			Ruoyu Zhang, \IEEEmembership{Senior Member,~IEEE},  
			Lei Yao,\\ 
			Wen Wu, 
			\IEEEmembership{Senior Member,~IEEE}
		}
		\vspace{-0.75em}
		\thanks{
			
			Chengzhi Ye, Ruoyu Zhang, Lei Yao and Wen Wu are with the Key Laboratory of Near-Range RF Sensing ICs and Microsystems (NJUST), Ministry of Education, School of Electronic and Optical Engineering, Nanjing University of Science and Technology, Nanjing 210094, China. Chengzhi Ye and Lei Yao are also with Qian Xuesen College, Nanjing University of Science and Technology, Nanjing 210094, China. (e-mail: qxsycz8166@njust.edu.cn; ryzhang19@njust.edu.cn; yaolei@njust.edu.cn; wuwen@njust.edu.cn).
		}
		\vspace{-2em}
	}

	
	\maketitle
	
	

	%
	\IEEEpeerreviewmaketitle

	\begin{abstract}
	Movable Antenna (MA) technology is emerging as a promising advancement with the potential to significantly enhance the performance of future wireless communication and sensing systems. In this paper, we address two-dimensional (2D) direction of arrival (DOA) estimation via joint shape-position optimization. Specifically, we formulate an optimization problem aimed at minimizing the Cramér-Rao Bound (CRB) based on a 2D DOA estimation model for MA systems. To tackle the highly non-convex nature of this CRB minimization, we investigate the spatial utilization of the movable region (MR) under minimum antenna spacing constraints. By demonstrating that an equilateral triangle yields the minimum overlap area, we strategically design an equilateral triangular MR. This specific geometric configuration enables the exploitation of structural symmetry to simplify the geometric constraints, which effectively reduces the complexity of solving the optimization problem. Subsequently, we derive the optimal MA positions by selecting the candidate locations farthest from the centroid of MR. The results demonstrate that the proposed joint shape-position optimization substantially enhances 2D DOA estimation performance.
	\end{abstract}
	\vspace{-0.5em}
	\vspace{-0.75em}
	\section{Introduction}

	Integrated sensing and communication (ISAC) has emerged as a core six-generation (6G) technology enabling simultaneous data transmission and environmental sensing \cite{ISACSurveyXingzhouWen2022,DingMAISAC2025}. A critical component of ISAC is accurate two-dimensional (2D) direction-of-arrival (DOA) estimation, which provides precise three-dimensional localization of diverse intelligent devices for emerging Internet of Things (IoT) applications, such as autonomous industrial monitoring and smart robots. However, existing DOA estimation methods primarily rely on conventional fixed-position antenna (FPA) arrays and increasing the number of antennas and RF chains for improving resolution, which is generally impractical for power- and space-constrained IoT applications \cite{RuoyuZhangDOAEstimation2022,11071308}.

	In recent years, movable antenna (MA) technology has emerged as a promising paradigm by allowing antennas to flexibly move within a confined region. 
	Compared to conventional FPA,  MA can achieve additional spatial degrees of freedom by synthesizing a superior aperture without increasing the number of antennas or RF chains, 
	thereby enhancing angular resolution the performance in wireless sensing and communication \cite{ChengzhiYeGPEMA2025,11456856, ShangMAPositionAngle2026}. For instance, the authors in \cite{maMovableAntennaEnhanced2024} optimized the element positions of a MA array to reduce the Cramér-Rao Bound (CRB) for DOA estimation, addressing the ambiguity issue in conventional FPA arrays. A tensor decomposition-based channel estimation method was proposed in \cite{zhangChannelEstimationMovableantenna2024}, which leverages MAs and a two-stage movable training mode to achieve channel estimation performance. 
	The aforementioned works mainly leverage the position variability of MAs in a plane to enhance communication and sensing performance. As a further step,
	the authors in \cite{6DMAWSensing2025} introduced a six-dimensional MA method, which jointly optimizes positions and rotations to minimize the CRB and realize high estimation performance. While existing studies \cite{maMovableAntennaEnhanced2024,ShangMAPositionAngle2026, zhangChannelEstimationMovableantenna2024,MAEnhancedZhuLipeng2024,6DMAWSensing2025} have extensively explored the optimization of MA positions and rotations, they predominantly focus on positioning within a given and fixed movable region (MR). However, in practice, the region shape is a critical engineering variable constrained by the platform’s limited size. By overlooking how the shape itself dictates the achievable array manifold, these conventional approaches fail to fully exploit the spatial degrees of freedom for optimizing estimation.
	
	To address the aforementioned limitations, we propose a novel joint shape-position optimization framework based on the CRB minimization for 2D DOA estimation. Instead of relying on predefined MR boundaries, we mathematically prove that an equilateral triangle yields the minimum overlap area under the minimum antenna spacing constraint. This finding directly motivates the design of an equilateral triangular movable region to maximize spatial utilization. Moreover, the structural symmetry of this geometric configuration can simplify the objective function and effectively reduce the complexity of solving the optimization problem. Building upon this simplified formulation, we develop an efficient strategy to determine the optimal antenna placement by selecting the candidate locations farthest from the MR centroid. Finally, the multiple signal classification (MUSIC) algorithm is executed on the optimized movable antenna array to achieve high-resolution direction finding. Extensive simulation results verify that the proposed design substantially outperforms conventional array configurations in terms of spatial resource utilization and estimation accuracy.

	In this section, we establish the system model for 2D DOA estimation. We consider a MR $\mathcal{C}$ with the area $\Theta$ containing $N$ MAs. Their coordinates in the Cartesian coordinate system $xOy$ are given by
	$\bm{P} = [ \bm{x}, \bm{y} ]^{\mathrm{T}}$,
	where $\bm{x} = [ x_1, \ldots, x_N ]^{\mathrm{T}}$ and $\bm{y} = [ y_1, \ldots, y_N ]^{\mathrm{T}}$ represent the $x$-axis and $y$-axis coordinates of the MAs, respectively. Furthermore, the distance between any two antennas is required to be greater than or equal to $d_{\min}$.
	Assume that $K$ far-field source signals impinge on this MA array. The received signal can be expressed as
	\begin{align}
		\bm{y}(t) = \sum\nolimits_{k=1}^{K} \bm{a}(\bm{r}_k) s_k(t) + \bm{n}(t),
	\end{align}
	where $s_k(t)$ denotes the $k$-th source signal, and $\bm{n}(t) \sim \mathcal{CN}(\bm{0}_{M \times 1}, \sigma_n^2 \bm{I}_{M})$ represents additive white Gaussian noise. $
	\bm{a}(\bm{r}_k) = [ 1, \mathrm{e}^{\mathrm{j}\frac{2\pi}{\lambda} r_{k,1}}, \ldots, \mathrm{e}^{\mathrm{j}\frac{2\pi}{\lambda} r_{k,N}} ]^{\mathrm{T}}$ denotes the $k$-th steering vector,
	where $r_{n,k} = x_n \vartheta_k + y_n \varphi_k$, with $\vartheta_k \triangleq \sin \theta_k \cos \phi_k$ and $\varphi_k \triangleq \cos \theta_k$.
	Considering $T$ snapshots, the received signal can be rewritten in the matrix form as
	\begin{align}
		\bm{Y} = \bm{A} \bm{S} + \bm{N},
	\end{align}
	where $\bm{Y} = [ \bm{y}(1), \bm{y}(2), \ldots, \bm{y}(T) ]\in \mathbb{C}^{N\times T} $, $\bm{S} = [ \bm{s}(1), \bm{s}(2), \ldots, \bm{s}(T) ] \in \mathbb{C}^{K \times T}$, $\bm{A} = [ \bm{a}(\bm{r}_1), \bm{a}(\bm{r}_2), \ldots, \bm{a}(\bm{r}_K) ]\in \mathbb{C}^{N \times K} $, and $\bm{N} = [ \bm{n}(1), \bm{n}(2), \ldots, \bm{n}(T) ]\in \mathbb{C}^{N \times T} $.
	
	Our objective is to efficiently estimate the DOA using the MA array. As indicated by the steering vector formulation, although the area of the MR is fixed, both the shape of the region and the positions of the MAs can be reconfigured. In the next section, we will derive the optimization of the region's shape and the MA positions by minimizing the CRB under the constraint of fixed area $\Theta$.
	\vspace{-0.75em}
	\section{Joint Shape-Position-Optimization}\label{III}
	
	In this section, we propose a joint shape-position optimization to minimize the CRB lower bound for optimizing the MA system, thereby achieving high-precision DOA estimation.
	
	It is well-known that the normalized RMSE serves as a metric for parameter estimation performance. For the $k$-th parameter, the relationship with the CRB satisfies
	$\mathrm{RMSE}(u_k) \geq \sqrt{\mathrm{CRB}(u_k)}.$
	Moreover, when a sufficient number of Monte Carlo trials are conducted, the RMSE closely approaches the CRB. Therefore, to achieve higher accuracy in DOA estimation, we must minimize the CRB as much as possible. In the following, we will sequentially address the optimization from both the shape and position perspectives.

	The CRB for 2D DOA estimation can be expressed as \cite{kay1993fundamentals}
	\begin{align}
		\mathrm{CRB}(\vartheta_k) & = Q /[{\operatorname{var}(\bm{x}) - {\operatorname{cov}^2(\bm{x},\bm{y})}/{\operatorname{var}(\bm{y})}}], \label{CRBu}\\
		\mathrm{CRB}(\varphi_k) & = Q /[{\operatorname{var}(\bm{y}) - {\operatorname{cov}^2(\bm{x},\bm{y})}/{\operatorname{var}(\bm{x})}}],\label{CRBv}
	\end{align}
	where $Q=\frac{\sigma_n^2 \lambda^2}{8\pi^2 T \sigma_k^2 N}$, $\sigma_k^2$ denotes the power of the $k$-th signal, $\operatorname{var}(\bm{x}) =\overline{x^2} - \bar{x}^2$, $\operatorname{var}(\bm{x}) =\overline{y^2} - \bar{y}^2$, $\operatorname{cov}(\bm{x},\bm{y}) = \overline{xy} - \bar{x}\cdot\bar{y}$, $\overline{x^2}= \sum_{n=1}^{N} x_n^2/N$, $\bar{x}=\sum_{n=1}^{N} x_n /N$, $\overline{y^2}= \sum_{n=1}^{N} y_n^2/N$, $\bar{y}=\sum_{n=1}^{N} y_n /N$, and $\overline{x_ny_n}=\sum_{n=1}^{N} x_n y_n/N$.
	\eqref{CRBu} and \eqref{CRBv} indicate that $\mathrm{CRB}(\vartheta_k)$ and $\mathrm{CRB}(\varphi_k)$ are functions of the MA coordinates, and they decrease as $\operatorname{var}(\bm{x}) - {\operatorname{cov}^2(\bm{x},\bm{y})}/{\operatorname{var}(\bm{y})}$ and $\operatorname{var}(\bm{y}) - {\operatorname{cov}^2(\bm{x},\bm{y})}/{\operatorname{var}(\bm{x})}$ increase, respectively. Therefore, to minimize $\mathrm{CRB}(\vartheta_k)$ and $\mathrm{CRB}(\varphi_k)$, we need to maximize these two terms. Consequently, the optimization objective can be transformed into
	\begin{align}
		\min \enspace \mathrm{CRB}(\vartheta) \enspace& \Leftrightarrow\enspace \max_{\bm{x},\bm{y}} \enspace \operatorname{var}(\bm{x}) - {\operatorname{cov}^2(\bm{x},\bm{y})}/{\operatorname{var}(\bm{y})}, \\ 
		\min \enspace \mathrm{CRB}(\varphi) \enspace& \Leftrightarrow\enspace \max_{\bm{x},\bm{y}} \enspace \operatorname{var}(\bm{y}) -{\operatorname{cov}^2(\bm{x},\bm{y})}/{\operatorname{var}(\bm{x})}.
	\end{align}
	To simultaneously optimize $\mathrm{CRB}(\vartheta_k)$ and $\mathrm{CRB}(\varphi_k)$, we consider both objective functions equally important and adopt an equal-weight linear combination. Additionally, all MAs must be placed within $\mathcal{C}$. Thus, the optimization problem can be formulated as
\begin{subequations}\label{maxP}
		\begin{align}
		\max_{\bm{x},\bm{y},\mathcal{C}}& \quad  \operatorname{var}(\bm{x}) + \operatorname{var}(\bm{y}) - \tfrac{\operatorname{cov}^2(\bm{x},\bm{y})}{\operatorname{var}(\bm{x})} - \tfrac{\operatorname{cov}^2(\bm{x},\bm{y})}{\operatorname{var}(\bm{y})},\\
		\mathrm{s.t.}& \quad\bm{P} \in \mathcal{C},
		\sqrt{(x_{n_1} - x_{n_2})^2 + (y_{n_1} - y_{n_2})^2} \ge d_{\min},\nonumber\\
		&\quad n_1, n_2=1,\ldots,N. \label{Const}
	\end{align}
\end{subequations}
		To maximize spatial utilization, MAs should be densely arranged to maximize MR space utilization. Theorem \ref{Theorem1} provides the minimum overlap area structure that satisfies the objective function.
	\begin{figure}
		\centering
		\includegraphics[width=0.9\linewidth]{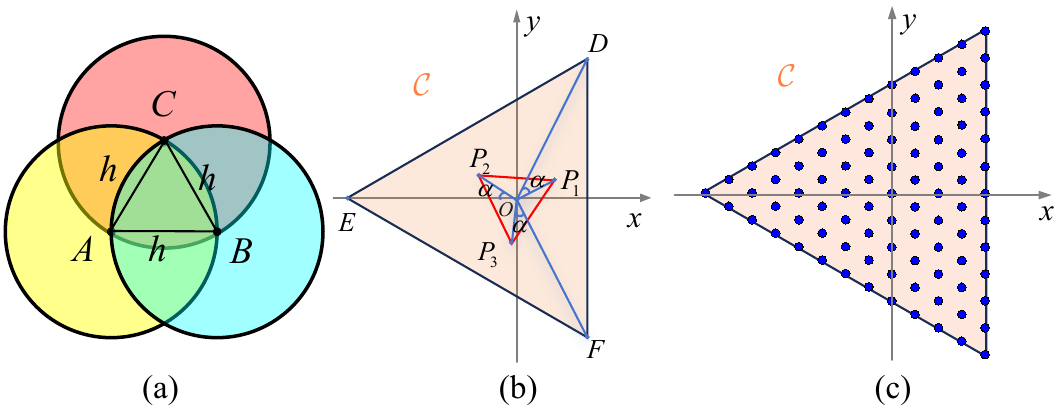}
		\caption{(a) Schematic illustration for the proof of Theorem \ref{Theorem1}. 
			(b) Schematic illustration for the proof of Theorem \ref{Theorem2}. 
			(c) Schematic illustration for the optimization of MA positions.}
		\label{proof1}
		\vspace{-1em}
	\end{figure}
	\begin{theorem}\label{Theorem1}
		Let $A$, $B$, and $C$ be three points on a plane with pairwise Euclidean distances of at least $h$. Let $\odot A$, $\odot B$, and $\odot C$ denote the radius $h$ centered at these points. The minimum total occupied area of these circles, which corresponds to the maximum overlap area among them, is achieved when $A$, $B$, and $C$ form an equilateral triangle with side length $h$.
	\end{theorem}
	\begin{proof}
		As illustrated in \Cref{proof1} (a), minimizing the total area occupied by $\odot A$, $\odot B$, and $\odot C$ is mathematically equivalent to maximizing their mutual overlapping area. The intersection between $\odot A$ and $\odot B$ reaches its maximum value when the distance between their centers is at the minimum allowable threshold such that $\overline{AB} = h$, where $\overline{AB}$ denotes the Euclidean distance. To further minimize the total occupied area, the third point $C$ must be positioned to maximize its joint overlap with both $\odot A$ and $\odot B$, which requires $C$ to be as close as possible to both $A$ and $B$ simultaneously. Given the distance constraints $\overline{AC} \ge h$ and $\overline{BC} \ge h$, the maximum shared area is achieved if and only if $\overline{AC} = \overline{BC} = h$, where these equalities represent the boundary of the feasible region. Consequently, the points $A$, $B$, and $C$ must form an equilateral triangle with side length $h$ to maximize the total intersection area and thus minimize the resultant occupied area.
	\end{proof}
	\vspace{-0.5em}
		Theorem \ref{Theorem1} describes the minimum overlap area structure that maximizes the utilization of spatial resources.
		Based on Theorem \ref{Theorem1}, we design an equilateral triangular MR that improves DOA estimation performance under the condition of equal area. Without loss of generality, let the center of the equilateral triangular MR be the origin. It should be noted that the objective function in \eqref{maxP} exhibits cyclic symmetry with respect to $\bm{x}$ and $\bm{y}$. Furthermore, the equilateral triangular MR also possesses a high degree of symmetry. Therefore, we can further simplify the optimization problem.
		\begin{theorem}\label{Theorem2}
			When the number of MAs is a multiple of three and the MR is an equilateral triangular MR, the optimization problem \eqref{maxP} can be simplified as
			\begin{subequations}
				\begin{align}
				\max_{\bm{\rho},\bm{x},\bm{y},\mathcal{C}} &\quad
				\overline{\rho^2}, \\
				\mathrm{s.t.} \quad&
				\bar{x}= \bar{y}=
				\overline{xy}=0, \enspace\eqref{Const}  \label{const3}.
			\end{align}
			\end{subequations}
			where $\bm{\rho} = [\rho_1, \ldots, \rho_N]^{\mathrm{T}}$, $\overline{\rho^2} =\tfrac{1}{N} \sum\nolimits_{n=1}^{N} \rho_n^2$, and $\rho_n=\sqrt{x_n^2+y_n^2}$.
	\end{theorem}
\vspace{-0.5em}
	\begin{proof}
		Based on the expression of $\operatorname{var}{(\bm{x})}$, $\operatorname{var}(\bm{y})$, $\operatorname{cov}(\bm{x},\bm{y})$, and $\rho_n^2=x_n^2+y_n^2$, the objective function \eqref{maxP} can be rewritten as
		\begin{align}\label{Max}
			& \overline{\rho^2} -\bar{x}^2-\bar{y}^2-\tfrac{(\overline{xy}-\bar{x}-\bar{y})^2}{\operatorname{var}(\bm{x})}-\tfrac{(\overline{xy}-\bar{x}-\bar{y})^2}{\operatorname{var}(\bm{y})}.
		\end{align}
		To maximize \eqref{Max}, we need to maximize $ \overline{\rho^2}$ while making $\bar{x}$, $\bar{y}$, and $\overline{xy}$ as close to zero as possible. As illustrated in Fig. \ref{Fig2}(b) and (c), the movable region is an equilateral triangle $DEF$ denoted as $\mathcal{C}$. This region is symmetric with respect to the $x$ axis, which inherently satisfies $\bar{y}=0$. Without loss of generality, we proceed to prove that $\bar{x}=0$ and $\overline{xy}=0$ for a given set of three MAs. For any MA positioned at a specific point $P_1$ such that the angle $\angle DOP_1 = \alpha$, there inevitably exist two corresponding MAs located at points $P_2$ and $P_3$. These two points are geometrically constrained such that the angle $\angle EOP_2 = \alpha$ with the radial distance $\overline{OP_2} = \overline{OP_1}$, and the angle $\angle FOP_3 = \alpha$ with the radial distance $\overline{OP_3} = \overline{OP_1}$.
		The coordinates of MAs $P_1$, $P_2$, and $P_3$ are therefore given by 
		\begin{align}
			P_1:&\bigl(\overline{OP_1}\cos (\tfrac{\pi}{3}-\alpha),\overline{OP_1}\sin \bigl(\tfrac{\pi}{3}-\alpha\bigr)\bigr)=(x_{P_1},y_{P_1}),\nonumber\\
			P_2:&\bigl(-\overline{OP_2}\cos \alpha,\overline{OP_2}\sin \alpha\bigr)=(x_{P_2},y_{P_2}),\\
			P_3:&\bigl(-\overline{OP_3}\cos (\alpha-\tfrac{\pi}{6}),-\overline{OP_3}\sin (\alpha-\tfrac{\pi}{6})\bigr)=(x_{P_3},y_{P_3}).\nonumber
		\end{align}
				\begin{figure*}
			\begin{align}\label{y=0}
				\sum\nolimits_{n=1}^{3}x_{P_n}
				&=\overline{OP_1}\cos \bigl(\tfrac{\pi}{3}-\alpha\bigr)
				-\overline{OP_2}\cos \alpha
				-\overline{OP_3}\cos \bigl(\alpha-\tfrac{\pi}{6}\bigr)=0.
				\\\label{xy}
				\sum\nolimits_{n=1}^{3}x_{P_n}y_{P_n}
				&=|\overline{OP_1}|^2\sin \bigl(\tfrac{\pi}{3}-\alpha\bigr)\cos \bigl(\tfrac{\pi}{3}-\alpha\bigr)-|\overline{OP_2}|^2\sin \alpha\cos \alpha
				-|\overline{OP_3}|^2\sin \bigl(\alpha-\tfrac{\pi}{6}\bigr)\cos \bigl(\alpha-\tfrac{\pi}{6}\bigr)\nonumber\\
				&=|\overline{OP_1}|^2\sin \bigl(\tfrac{2\pi}{3}-2\alpha\bigr)
				-|\overline{OP_2}|^2\sin 2\alpha
				-|\overline{OP_3}|^2\sin \bigl(2\alpha-\tfrac{\pi}{3}\bigr)=0.
			\end{align}
			\hrule
		\end{figure*}
		Therefore, the expression of $\bar{x}$ and $\overline{xy}$ can be given by \eqref{y=0} and \eqref{xy}, where the second equality in \eqref{xy} is obtained by applying the product-to-sum trigonometric identity.  
		Because the selection of the MA $P_1$ is arbitrary, one can always find two additional MAs such that $\bar{x}=0$ and $\overline{xy}=0$ are satisfied. 
		For a system comprising $N$ antennas, we can partition the array into $N/3$ distinct subsets. Within each individual subset, the local spatial coordinates strictly satisfy $\bar{x}=0$, $\bar{y}=0$, and $\overline{xy}=0$. 
		To this end, we completes the proof of Theorem \ref{Theorem2}.
	\end{proof}
	Theorem \ref{Theorem2} demonstrates that the equilateral triangular MR perfectly satisfies all the conditions in \eqref{const3} and obeys Theorem \ref{Theorem1}, thereby both maximizing the objective function \eqref{maxP} and take full advantage of spcace sources.

	In the next step, we elaborate on how to optimize the positions of MAs within an equilateral triangular MR. As illustrated in \Cref{proof1} (c), we consider the case $l=md_{\min}$ without loss of generality, where $l$ denotes the MR length and $m $ represents a positive integer, as this configuration aligns with practical antenna array designs.The equilateral triangular MR is discretized into a set of candidate positions $\mathcal{B}$ with the minimum distance $d_{\min}$. These candidate positions correspond to the blue dots illustrated in \Cref{proof1}(c). This discrete set contains $m(m+1)/2$ points in total and represents the minimum overlap area that satisfies the constraints of the optimization problem. For any given number of antennas $N \leq |\mathcal{B}|$, the optimal positions for the $N$ MAs are selected from the candidate set $\mathcal{B}$. These positions specifically correspond to the $N$ candidates located farthest from the centroid of the triangular region. The overall procedure of the proposed method is summarized in Algorithm \ref{Algorithm1}. 

\begin{algorithm}[t]
	\renewcommand{\algorithmicrequire}{\textbf{Input}} 
	\renewcommand{\algorithmicensure}{\textbf{Output}} 
	\caption{Proposed Joint Shape-Position Optimization Algorithm}
	\label{Algorithm1}
	\begin{algorithmic}[1]
		\Require 
		Length of the MR $l$, number of MAs $N$, and the minimum distance $d_{\min}$.
		\State Calculate the positive integer $m = l/d_{\min}$ and construct the discrete candidate position set $\mathcal{B}$ within the equilateral triangular MR;
		\State Compute the spatial distance from each candidate point in $\mathcal{B}$ to the centroid of the triangular region;
		\State Sort all candidate points in descending order according to their computed distances to the centroid;
		\State Select the position coordinates corresponding to the first $N$ largest distance values;
		\Ensure Optimal positions of the MAs.
	\end{algorithmic}
\end{algorithm}

	\section{Simulation Results}
	\label{Sec:Simulation}
	
	In this section, we evaluate the performance of the proposed MR for MA through numerical simulations based on the MUSIC algorithm. 	
	We set the parameters as follows, including the MR $\Theta=    27.71\lambda^2$, number of MAs $M=36$, number of snapshots $T=1$, and the number of the source signal is $K=1$ with $\theta=45^\circ,\phi=60^\circ$.
	The length of the proposed equilateral triangular MR is $8\lambda$,  the length of the square area is $5.26\lambda$, and the radius of the circle area is $2.86\lambda$, respectively.
	The labels \textbf{PMA}, \textbf{SMA}, \textbf{UCA}, and \textbf{URA} denote the proposed MA (PMA) array, square area-based MA (SMA) array, uniform circular array and uniform rectangular array, respectively. The labels \textbf{PMA-CRB}, \textbf{SMA-CRB}, \textbf{UCA-CRB}, and \textbf{URA-CRB} denote the CRBs, respectively.

	\begin{figure*}[t]
		\centering
		\subfloat[PMA.]{
			\includegraphics[width=0.22\textwidth]{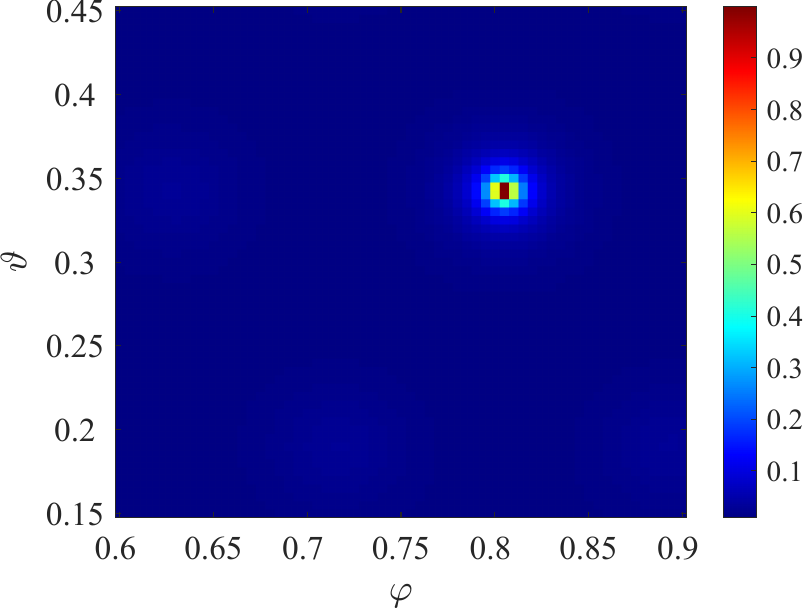}
			\label{Fig_PMA}
		}
		\hfill
		\subfloat[SMA.]{
			\includegraphics[width=0.22\textwidth]{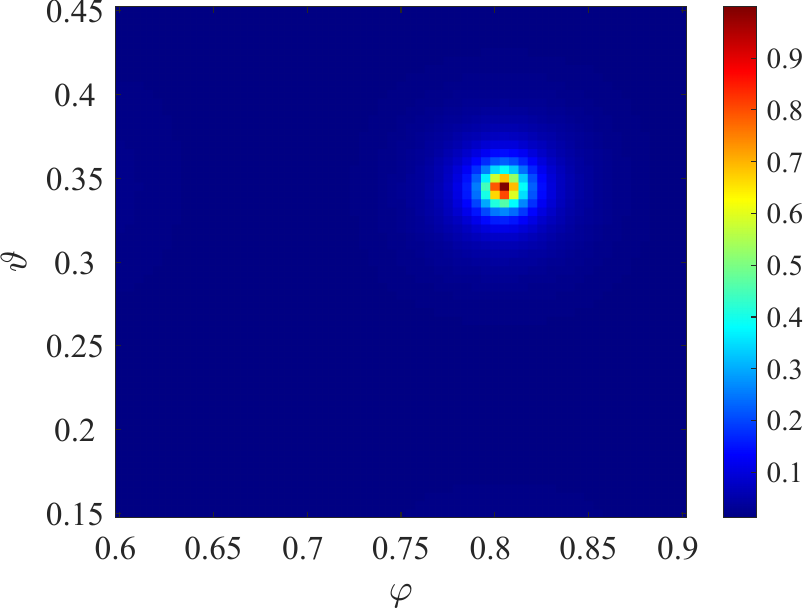}
			\label{Fig_SMA}
		}
		\hfill
		\subfloat[UCA.]{
			\includegraphics[width=0.22\textwidth]{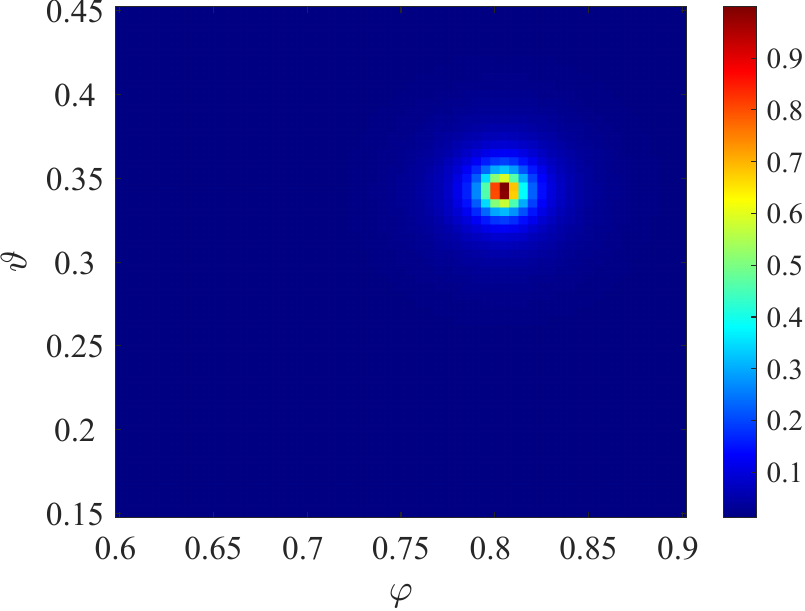}
			\label{Fig_UCA}
		}
		\hfill
		\subfloat[URA.]{
			\includegraphics[width=0.22\textwidth]{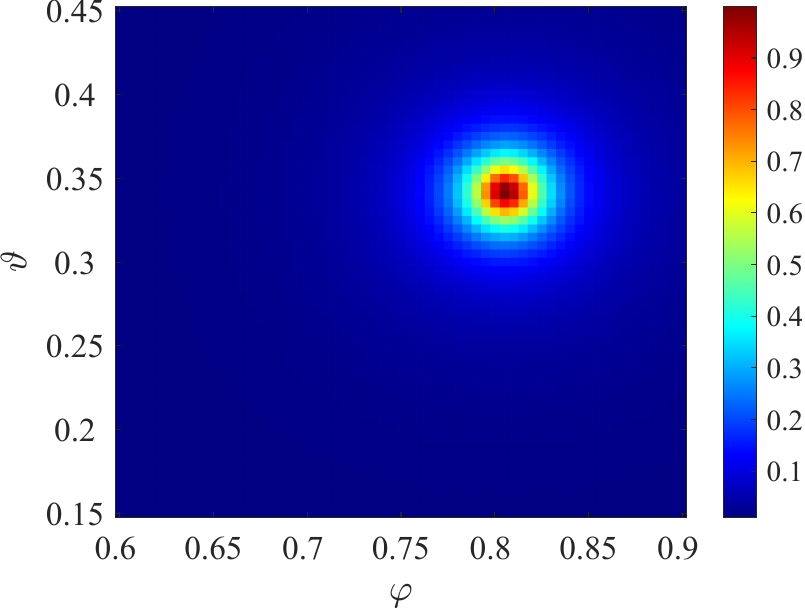}
			\label{Fig_URA}
		}
		\caption{Power spectrum function versus $\vartheta$ and $\varphi$ for different arrays.}
		\label{Fig_spectrum_comparison}
		\label{Fig2}
	\end{figure*}

	\begin{figure}[t]
		\centering
		\begin{minipage}{0.45\linewidth}
			\centering
			\includegraphics[width=\linewidth]{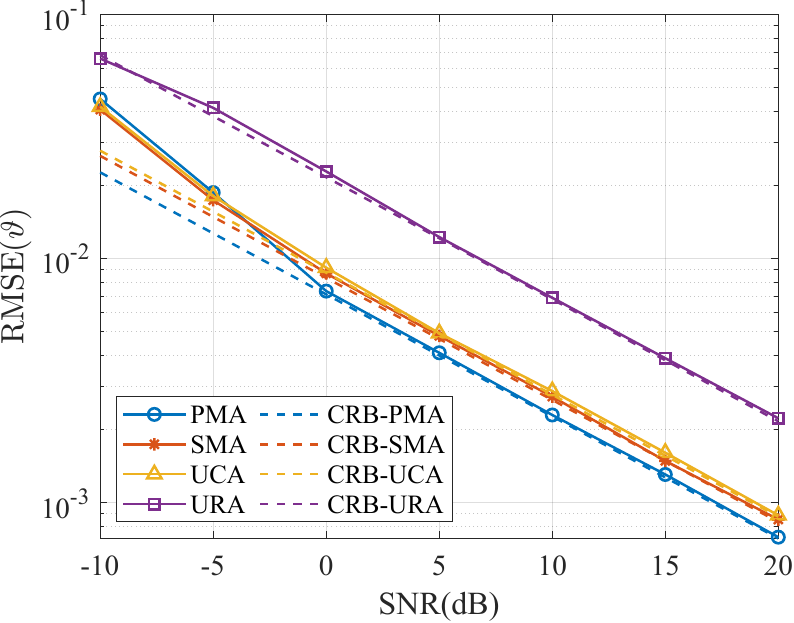}
			\caption{RMSE($\vartheta$) versus SNR.}
			\label{Fig_vartheta}
		\end{minipage}
		\hfill
		\begin{minipage}{0.45\linewidth}
			\centering
			\includegraphics[width=\linewidth]{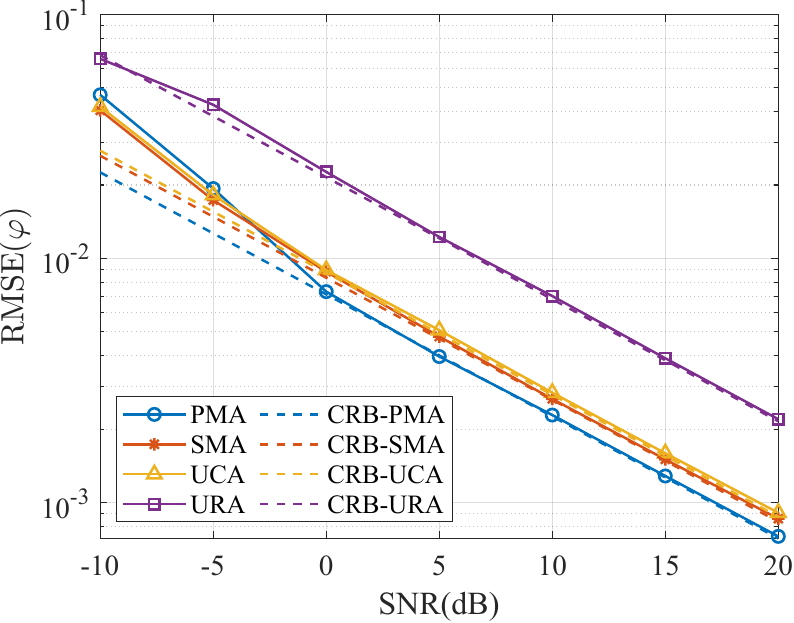}
			\caption{RMSE($\varphi$) versus SNR.}
			\label{Fig_varphi}
		\end{minipage}
	\end{figure}
	\begin{figure}[t]
		\centering
		\begin{minipage}{0.45\linewidth}
			\centering
			\includegraphics[width=\linewidth]{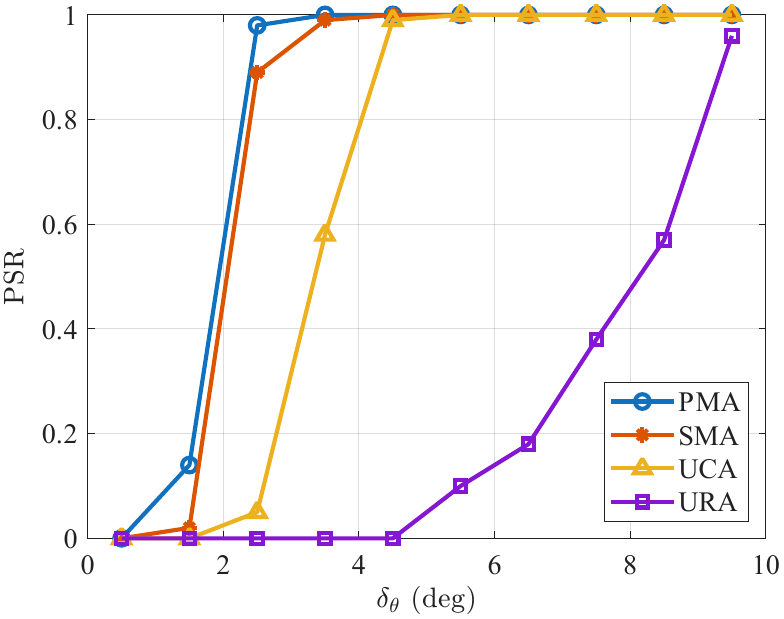}
			\caption{PSR versus $\delta_\theta$.}
			\label{thetaseparation}
		\end{minipage}
		\hfill
		\begin{minipage}{0.45\linewidth}
			\centering
			\includegraphics[width=\linewidth]{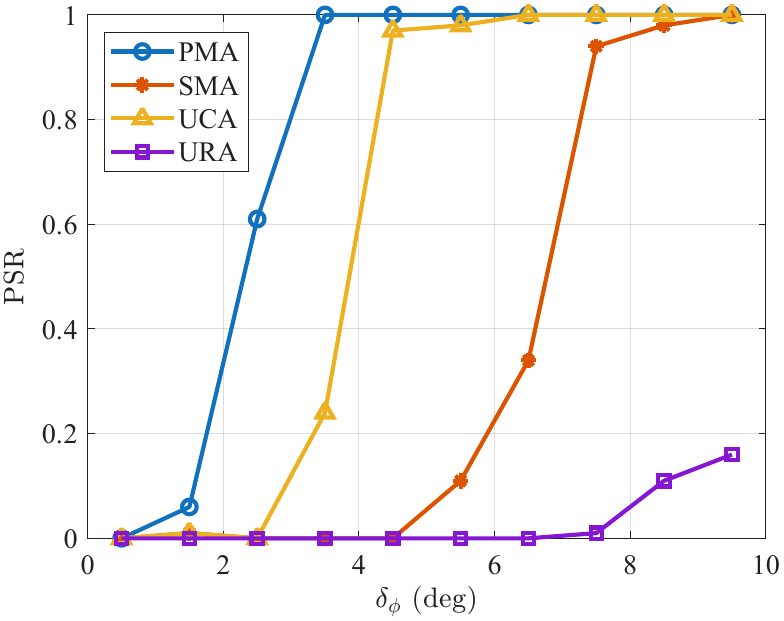}
			\caption{PSR versus $\delta_\phi$.}
			\label{phiseparation}
		\end{minipage}
		\vspace{-1em}
	\end{figure}
	\begin{figure}[t]
		\centering
		\begin{minipage}{0.45\linewidth}
			\centering
			\includegraphics[width=\linewidth]{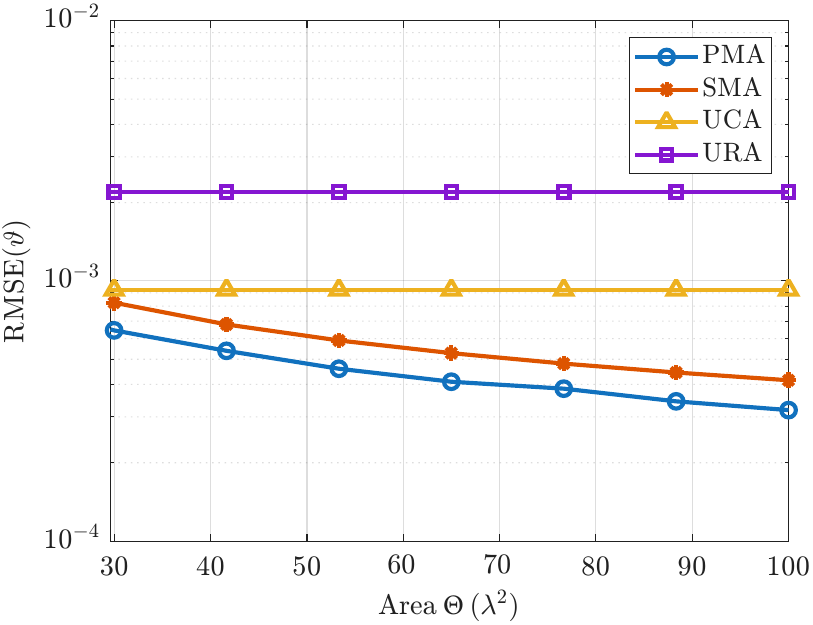}
			\caption{RMSE($\vartheta$) versus $\Theta$.}
			\label{Atheta_SNRb}
		\end{minipage}
		\hfill
		\begin{minipage}{0.45\linewidth}
			\centering
			\includegraphics[width=\linewidth]{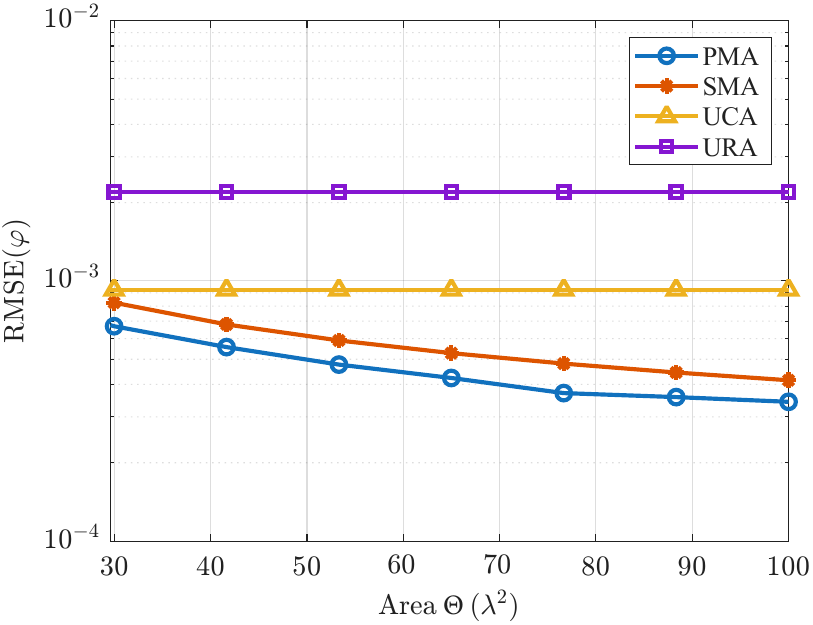}
			\caption{RMSE($\varphi$) versus $\Theta$.}
			\label{ARMSE_varphib}
		\end{minipage}
	\end{figure}

	\Cref{Fig2} illustrate the power spectrum function with respect to the grids of $\vartheta$ and $\varphi$ during spectral peak search. The SNR is set to 20 dB. It can be observed that the proposed array exhibits a narrower main lobe compared to other arrays, resulting in a sharper peak. This characteristic reduces the probability of errors in DOA estimation, thereby demonstrating the superior performance of the proposed array.
	
	\Cref{Fig_vartheta,Fig_varphi} depict the RMSE of $\vartheta$ and $\varphi$ as a function of SNR. When SNR $< 0$ dB, the RMSEs of PMA, SMA, and UCA are significantly higher than the CRB, which is primarily attributed to the influence of grating lobes. Conversely, when SNR $> 0$ dB, the effect of grating lobes is suppressed by the increasing SNR, causing the estimation performance of all arrays to converge to the CRB. The proposed array achieves the smallest CRB, indicating the best estimation performance among the compared arrays.
	
	\Cref{thetaseparation,phiseparation} present the probability of successful resolution (PSR) for $\theta$ and $\phi$, respectively. We set three sources, where one is set as $\theta_1=135^\circ, \phi=115^\circ$, and two are closely spaced from the direction $\theta_2=45^\circ, \phi_2=60^\circ$ and $\theta_3=45^\circ+\delta_\theta, \phi_3=60^\circ+\delta_\phi$. $\delta_\theta=0$ in \Cref{thetaseparation} and $\delta_\phi=0$ in \Cref{phiseparation}, respectively. The evaluated methods are identified to successfully distinguish the two sources if $|\hat{\theta}_{l,k}-\theta_k|<\delta_\theta/2$ in \Cref{thetaseparation} and $|\hat{\phi}_{l,k}-\phi_k|<\delta_\phi/2$ in \Cref{phiseparation}, respectively. For the URA, a high PSR in estimating $\theta$ and $\phi$ is achieved only when the angular separation exceeds $10^\circ$. In contrast, the PMA can successfully estimate these parameters even when the separation is as small as $3^\circ$. Furthermore, while the SMA and UCA exhibit higher estimation accuracy than the URA, they remain inferior to the PMA.
	
	\Cref{Atheta_SNRb,ARMSE_varphib} demonstrate the performance of $\vartheta$ and $\varphi$ versus the area of MR. The SNR is set as $10$ dB. Regarding both $\vartheta$ and $\varphi$, the performance of conventional UCA and URA remains invariant as the MR size increases. In contrast, the SMA and PMA exhibit monotonous improvements. Furthermore, the PMA consistently outperforms the SMA across the entire range of MR dimensions.

	\section{Conclusion}
	
	This paper addressed 2D DOA estimation using MA technology by minimizing the CRB via joint shape-position optimization. Specifically, we mathematically proved that an equilateral triangular region yields the minimum overlap area to maximize spatial utilization. Moreover, the structural symmetry of this geometry simplified the optimization problem. Subsequently, we determined the optimal antenna positions by selecting the candidate locations farthest from the centroid. Simulations demonstrated that the proposed design significantly enhanced estimation performance by reducing the main lobe bandwidth during spectral peak search.

	
	
	
	%
	\bibliographystyle{IEEEtran}
	\bibliography{mybib_Abbreviation_1}

\end{document}